\documentclass[conference]{IEEEtran}
\IEEEoverridecommandlockouts
\usepackage{cite}
\usepackage{amsmath,amssymb,amsfonts,amsthm}
\usepackage{graphicx}
\usepackage{textcomp}
\usepackage{xcolor}
\usepackage{tikz}
\usetikzlibrary{positioning}
\usepackage{cite,url,subcaption,multicol}

\usetikzlibrary{decorations.pathreplacing}

\def\BibTeX{{\rm B\kern-.05em{\sc i\kern-.025em b}\kern-.08em
    T\kern-.1667em\lower.7ex\hbox{E}\kern-.125emX}}
    
\usepackage[ruled]{algorithm}
\usepackage[noend]{algpseudocode}

\algdef{SE}{Upon}{EndUpon}[1]{\textbf{upon}\ #1\ \algorithmicdo}{}
\algtext*{EndUpon}

\theoremstyle{plain}
\newtheorem{theorem}{Theorem}
\newtheorem{corollary}{Corollary}
\newtheorem{lemma}[theorem]{Lemma}
\newtheorem{claim}[theorem]{Claim}

\theoremstyle{definition}
\newtheorem{definition}{Definition}

\theoremstyle{remark}
\newtheorem{remark}{Remark}

\newcommand{\qinzi}[1]{{\color{blue} #1}}

\newcommand{\ddegree}[1]{$(#1)$-dynaDegree}

\DeclareMathOperator{\E}{\mathcal{E}}

\DeclareMathOperator{\calB}{\mathcal{B}}
\DeclareMathOperator{\calH}{\mathcal{H}}

\newcommand{\range}{\operatorname{range}}
\newcommand{\interval}{\operatorname{interval}}
\newcommand{\calA}{\mathcal{A}}

\DeclareMathOperator{\low}{low}
\DeclareMathOperator{\high}{high}
\DeclareMathOperator{\pend}{end}

\newcommand*{\affmark}[1][*]{\textsuperscript{#1}}

\begin{document}

\title{Fault-tolerant Consensus \\in Anonymous Dynamic Network 
}

\author{\IEEEauthorblockN{Qinzi Zhang\affmark[1] and Lewis Tseng\affmark[2]*\thanks{*This material is based upon work partially supported by the National Science Foundation under Grant CNS-2238020.}}

\IEEEauthorblockA{
\affmark[1]Boston University, Boston, USA\\
\affmark[2]Clark University, Worcester, USA\\
E-mails: \protect \affmark[1]qinziz@bu.edu, 
\affmark[2]lewistseng@acm.org}
}

\maketitle

\begin{abstract}
This paper studies the feasibility of reaching consensus in an anonymous dynamic network. In our model, $n$ \textit{anonymous} nodes proceed in synchronous rounds. We adopt a hybrid fault model in which up to $f$ nodes may suffer crash or Byzantine faults, and the  \textit{dynamic} message adversary chooses a communication graph for each round. 

We introduce a \textit{stability} property of the dynamic network -- \textit{\ddegree{T,D}} for $T \geq 1$ and $n-1 \geq D \geq 1$ -- which requires that for every $T$ consecutive rounds, any fault-free node must have incoming directed links from at least $D$ distinct neighbors. These links might occur in different rounds during a $T$-round interval. \ddegree{1,n-1} means that the graph is a complete graph in every round. \ddegree{1,1} means that each node has at least one incoming neighbor in every round, but the set of incoming neighbor(s) at each node may change arbitrarily between rounds.

We show that exact consensus is impossible even with \ddegree{1,n-2}. For an arbitrary $T$, we show that for crash-tolerant approximate consensus, \ddegree{T, \lfloor n/2 \rfloor} and $n > 2f$ are together necessary and sufficient, whereas for Byzantine approximate consensus, \ddegree{T, \lfloor (n+3f)/2 \rfloor} and $n > 5f$ are together necessary and sufficient.

\end{abstract}

\begin{IEEEkeywords}
Consensus, Approximate consensus, Byzantine, Impossibility, Message adversary, Dynamic network
\end{IEEEkeywords}

\section{Introduction}

A dynamic network is a natural model for mobile devices equipped with wireless communication capability. Node mobility and unpredictable wireless signal (e.g., due to interference and attenuation) make the systems of mobile devices inherently dynamic. For example, nodes may join, leave, and move around, and communication links between nodes may appear and disappear over time. 

Recent works have studied consensus and other distributed tasks in various formulations of dynamic networks. Following \cite{Kuhn_dynamic_STOC10,Kuhn_dynamic_SIGACTNews11,Kunh_dynamicConsensus_PODC11,Nowak_tight_bound_asymptotic_JACM21,Nowak_Averaging_ICALP15}, we consider a network that does \textit{not} eventually stops changing. That is, there exists a dynamic message adversary that controls and chooses the set of communication links continually. More concretely, we study computability of fault-tolerant consensus in  the anonymous dynamic network model \cite{DiLuna_DynamicAnonymous_DISC23,DiLuna_DynamicAnonymous_PODC15,DiLuna_DynamicAnonymous_PODC23,DiLuna_DynamicAnonymous_OPODIS15} when nodes may crash or become Byzantine.  

\subsection{Anonymous Dynamic Network Model}

We consider a fixed set of $n$ nodes that proceed in \textit{synchronous} rounds and communicate by a \textit{broadcast} primitive. For example, such a broadcast primitive might be implemented by a medium access control (MAC) protocol in a wireless network. In each round, the communication graph is chosen by the dynamic message adversary.  We do \textit{not} assume the existence of a neighbor-discovery mechanism or an acknowledgement from a MAC layer; thus, nodes do not know the set of nodes that received their messages. 

In every round $t$, the adversary first chooses the directed links that are ``reliable'' for round $t$. In other words, the links that are \textit{not} chosen will drop any messages that were sent via them in round $t$. The adversary may use nodes' internal states at the beginning of the round and the algorithm specification to make the choice. We consider \textit{deterministic} algorithms in which nodes update internal states and generate message following the specification deterministically. 

The nodes do not know which links were chosen by the adversary. Each message is then delivered to the sender's outgoing neighbors, as defined by the links chosen by the adversary. After messages are received from incoming neighbors, the nodes transition to new states, and enter round $t+1$. Nodes are assumed to know $n$, the size of the network, and $f$, the upper bound on the number of node faults. 

Compared to prior works on dynamic networks, our model are different from two following perspectives:

\begin{itemize}
    \item \textit{Anonymity}: Nodes are assumed to be anonymous \cite{Angluin_anonymous_STOC80,Aspnes_Anonymous_PopulationProtocol_OPODIS05,Ruppert_HundredsImpossible_DC03,Tseng_PODC22_MAC,AbstractMAC_DISC2009,DiLuna_DynamicAnonymous_DISC23,DiLuna_DynamicAnonymous_PODC23}. That is, nodes are assumed to be identical and nodes do \textit{not} have a unique identity. Instead, nodes may use ``local''  communication ports or rely on the underlying communication layer (e.g., MAC layer) to distinguish messages received from different incoming neighbors. 

    \item \textit{Hybrid faults}: In addition to the message adversary, up to $f$ nodes may crash or have Byzantine behavior.


\end{itemize}

We are interested in the formulation, because most mobile devices are fragile and it is important to tolerate node faults. Moreover, in practice, it it not always straightforward to bootstrap a large-scale dynamic system so that all the nodes have unique and authenticated identity. For example, even though typical MAC protocols assume  unique MAC address, it is a good engineering and security practice \textit{not} to rely on this information in the upper layer.\footnote{There is also a privacy aspect. For example, recent versions of Android and iOS support MAC randomization to prevent from device tracking.} Plus, MAC spoofing is a well-known attack, which will add another layer of complexity in handling identities when designing Byzantine-tolerant algorithms.

This paper aims to answer the following question:

\vspace{10pt}
\hspace{20pt}
\fbox{%
  \parbox{180pt}{%
    
      When is consensus solvable in our anonymous dynamic network model? 
    
  }%
}

\vspace{5pt}

In this paper, we focus on the case when the nodes are able to communicate directly with each other, i.e., a single-hop network. The case of multi-hop communication is left as an interesting future work. Our model is motivated by the applications in drones, robots, and connected vehicles. In this applications, a team of fixed number of nodes are configured to communicate with each other, via dynamic wireless networks, in order to solve certain tasks (e.g., search, dynamic speed configuration, and flocking). In these applications, consensus is a key enabling primitive. 

\subsection{A Stability Property}

Towards our goal, we introduce a \textit{stability} property called ``\textit{\ddegree{T,D}}'' for $T \geq 1$ and $n-1 \geq D \geq 1$, which quantifies the number of distinct incoming neighbors for each fault-free node over any $T$-round interval.\footnote{``dynaDegree'' stands for dynamic degree.} $T$ is assumed to be finite, and both $T$ and $D$ are \textit{unknown} to nodes. We only use these parameters for analysis. We do not assume the graph contains self-loop; hence the parameter $D$ is upper bounded by $n-1$. 

More concretely, \ddegree{T,D} requires that for every $T$ \textit{consecutive} rounds in a given execution, any fault-free node must have incoming links from at least $D$ \textit{distinct} neighbors. These directed links might occur in different rounds during a $T$-round interval. \ddegree{1,n-1} means that the graph is a complete graph in every round. \ddegree{1,1} means that each node has at least one incoming neighbor in every round, but the incoming neighbor(s) may change arbitrarily between rounds. Figure \ref{fig:dynaDegree} presents an example execution in a $3$-node network that satisfies \ddegree{2,1}, but does not satisfy \ddegree{1,1}. 

\subsection{Contributions}

We characterise the feasibility of solving consensus in our model using the stability property. In particular, we have the following contributions: 

\begin{itemize}
    \item We identify how \ddegree{T,D} is related to a recently identified impossibility result by Gafni and Losa \cite{Gafni_TimeNotHealer_SSS23}, which then implies that exact consensus is impossible with \ddegree{1,n-2} even when no node may crash ($f=0$). In other words, there exists an execution such that the network satisfies the stability property \ddegree{1,n-2}, yet nodes are not able to agree on exactly the same output. 

    \item We identify that for crash-tolerant approximate consensus, \ddegree{T, \lfloor n/2 \rfloor} and $n > 2f$ are together necessary and sufficient. 

    \item For Byzantine approximate consensus, \ddegree{T, \lfloor (n+3f)/2 \rfloor} and $n > 5f$ are together necessary and sufficient.

\end{itemize}  
Section \ref{sec:crash} presents our crash-tolerant approximate consensus algorithm. Section \ref{sec:byz} presents our Byzantine approximate consensus algorithm. The impossibility results are presented in Section \ref{s:impossible}, which imply the necessity of the identified conditions. We conclude and discuss extensions in Section \ref{s:discussion}. 

\section{Preliminaries}
\label{s:preliminaries}

\subsection{Our Model: Anonymous Dynamic Network}

We consider a \textit{synchronous} message-passing system consisting of $n$ \textit{anonymous} nodes. Nodes only know $n$ and do \textit{not} have unique identities. For presentation and analysis purpose, we denote the set of nodes as the set of IDs, i.e., $\{1, \dots, n\}$. For brevity, we often denote it by $[n]$. 

\vspace{3pt}
\noindent\textbf{Node Faults.} ~~
We assume that at most $f$ nodes may become faulty. We consider both crash and Byzantine faults.
In the former model, a faulty node may crash and stop execution at any point of time. The latter model assumes faulty nodes may have an arbitrary faulty behavior, including sending different messages to different nodes. The set of faulty nodes is denoted by $\calB$. Nodes that are not faulty are called \textit{fault-free}. The set of fault-free nodes is denoted by $\calH$. 

\vspace{3pt}
\noindent\textbf{Communication and Message Adversary.} ~~
The underlying communication network is modeled as a synchronous dynamic network represented as a dynamic graph $G = (V, E)$, where $V$ is a static set of nodes $[n]$, and $\mathcal{E} : \mathbb{N} \rightarrow \{(u, v)~|~ (u,v) \in E \}$ is a function mapping a round number $t \in \mathbb{N}$ to a set of \textit{directed} links $\E(t)$. For $(u, v) \in \E(t)$, $v$ is said to be $u$'s outgoing neighbor and $u$ is said to be $v$'s incoming neighbor in round $t$. 

In round $t$, only messages sent over $\E(t)$ are delivered. All other messages are lost. We consider a dynamic message adversary that chooses $\E(t)$ in every round $t$. Note that there are different ways of modeling a dynamic network, e.g., using temporal graphs \cite{Nowak_Averaging_ICALP15,Nowak_tight_bound_asymptotic_JACM21}. We adopt the definition from \cite{Kuhn_dynamic_STOC10,Kuhn_dynamic_SIGACTNews11,Kunh_dynamicConsensus_PODC11}. 

We do not assume the existence of self-loop in $E$; however, nodes have the ability to send a message to itself. Such a message delivery \textit{cannot} be disrupted by the message adversary. In other words, a message sent to oneself is always delivered reliably. 

Most prior works assume that each node sends only one message and each message is of size at most $O(\log n)$ bits \cite{Kuhn_dynamic_STOC10,Kuhn_dynamic_SIGACTNews11,Kunh_dynamicConsensus_PODC11,Nowak_Averaging_ICALP15,Nowak_tight_bound_asymptotic_JACM21}. We adopt the same assumption of limited bandwidth of each edge. Section \ref{s:discussion} briefly discusses when each link has different bandwidth constraints.

\vspace{3pt}
\noindent\textbf{Anonymity and Port Number.} ~~Nodes execute the same code, because they are identical and the only difference is a potentially distinct input given to each node. They communicate with each other via a broadcast primitive. The delivery of messages are determined by the edge set chosen by the message adversary in each round $t$, namely $\E(t)$. 

Following \cite{Angluin_anonymous_STOC80}, we assume that each node has a ``local'' label for each incoming neighbor, i.e., a \textit{unique port number} for each incoming link. The labels are \textit{local} in the sense that two different nodes may use two different ports to identify messages received from the same node; therefore, it is not possible to use such information to assign global unique identities to all nodes in our model, without using consensus. 

The nodes, however, can use ports to distinguish the sender for each received message locally. Since the ports are assumed to be static throughout the execution of the algorithm, upon receipt of two incoming messages $m_1$ and $m_2$, a node $i$ is able to identify that these two messages are from two different senders by using the port numbers. In addition, a node $i$ has the ability to keep track of all the past messages received from a specific port. 

\begin{figure*}[t]
  \centering
  \begin{subfigure}[b]{0.45\textwidth}
    \centering
    \begin{tikzpicture}
  \node[circle, draw, minimum size=1cm] (circle1) at (0,0) {1};
  \node[circle, draw, minimum size=1cm] (circle2) at (2,1) {2};
  \node[circle, draw, minimum size=1cm] (circle3) at (4,0) {3};

  
\end{tikzpicture}
    \caption{When $t$ is odd, $\E(t)$ is empty}
    \label{fig:sub1}
  \end{subfigure}
  \begin{subfigure}[b]{0.45\textwidth}
    \centering
    \begin{tikzpicture}
  \node[circle, draw, minimum size=1cm] (circle1) at (0,0) {1};
  \node[circle, draw, minimum size=1cm] (circle2) at (2,1) {2};
  \node[circle, draw, minimum size=1cm] (circle3) at (4,0) {3};

  \draw[<->] (circle1) -- (circle2);
  \draw[<->] (circle2) -- (circle3);
  
\end{tikzpicture}
    \caption{When $t$ is even, $\E(t) = \{(1,2), (2, 1), (2, 3), (3, 2)\}$}
    \label{fig:sub2}
  \end{subfigure}
  \caption{\textbf{Illustration of an example message adversary.} Figure \ref{fig:sub1} shows that during odd rounds, the message adversary removes all the links, whereas Figure \ref{fig:sub2} shows that during even rounds, the adversary removes two links $(1,3)$ and $(3,1)$.} 
  \label{fig:dynaDegree}
\end{figure*}
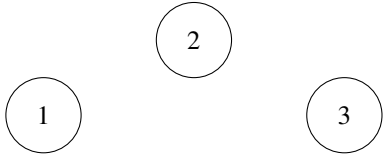
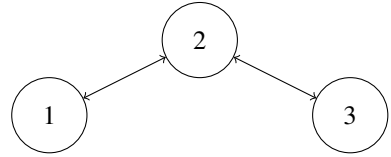

Formally, for node $i$, a port numbering is a bijection function $P_i : \{v \in V\} \rightarrow \{1, 2, \dots, n\}$. Recall that we assume each node knows $n$; hence, $P_i$ is well-defined. For two different nodes $i$ and $j$, $P_i$ may be different from $P_j$. Depending on the graph $G$ and the message adversary, nodes might never receive messages from a specific port. 

We assume that the underlying communication layer is authenticated in the sense that a Byzantine sender cannot tamper with the port numbering at fault-free nodes. As a result, Byzantine sender is \textit{not} able to send bogus messages on behalf of other nodes.

\subsection{A Stability Property: \ddegree{T,D}}

We characterize graphs using the following property:

\begin{definition}[\textbf{\ddegree{T,D}}]
\label{def:ddegree}
    \textit{A dynamic graph $G = (V, E)$ satisfies \ddegree{T,D} for a finite $T \geq 1$ and $n-1 \geq D \geq 1$ if for all $t \in \mathbb{N}$, in the \textit{static} graph $G_t := (V, \cup_{i = t}^{t+T-1} \E(t))$, the number of distinct incoming neighbors ``aggregated'' over any $T$-round interval at any fault-free node $i$ is at least $D$.}
\end{definition}

On a high-level, $G$ is the ``base'' communication graph that defines the capability for sending messages when all the links are reliable. In each round $t$, the message adversary chooses the set of reliable links, defined as $\E(t)$ -- links in $\E(t)$ deliver message reliably. Nodes do not know $G$, nor $\E$.

The property \ddegree{T,D} counts the number of distinct incoming neighbors $D$ over any $T$-round period. $G_t$ is a static graph whose set of links does not change over time. The set of links in $G_t$ is a union of $\E(t), \E(t+1), \cdots, \E(t+T-1)$. Any fault-free nodes in $G_t$ must have $\geq D$ distinct incoming neighbors. Nodes do \textit{not} know $T$ nor $D$. These parameters are only used for analysis. 

By definition, the incoming neighbors required by the $T$-interval dynamic degree does \textit{not} need to be fault-free. For example, the message adversary may choose links to deliver messages from Byzantine neighbors. 

\vspace{3pt}
\noindent\textbf{Example.}~~Consider Figure \ref{def:ddegree}. $G(V,E)$ has $V = \{1, 2, 3\}$ and $E = \{(1,2), (1, 3), (2, 1), (2, 3), (3, 1), (3, 2)\}$. 
The example satisfies \ddegree{2,1} because in any $2$-round interval, each node has at least 1 incoming neighbor. It does not satisfy \ddegree{1, 1}, because in odd rounds, the graph is disconnected.


\vspace{3pt}
\noindent\textbf{Comparison with Prior Stability Properties.}~~Prior papers identified several different stability properties on dynamic graphs for various distributed tasks. We compare \ddegree{T, D} with most relevant ones:

\begin{itemize}
    \item \textit{$T$-interval connectivity} \cite{Kuhn_dynamic_STOC10}: it requires that for every $T$ consecutive rounds, there exists a stable connected spanning subgraph. The set of dynamic links are assumed to be \textit{bi-directional}. In other words, the graph is strongly connected in every $T$-round internal.

    \item \textit{Rooted spanning tree} \cite{Nowak_Averaging_ICALP15,Nowak_tight_bound_asymptotic_JACM21,Kyrill_dynamicConsensus_ITCS23}: it requires that in every round $t$, the graph contains a \textit{directed} rooted spanning tree, i.e., there exists at least one ``coordinator'' that can reach every other node (potentially via a multi-hop route) in the graph for round $t$. 
\end{itemize}

Our condition is different from $T$-interval connectivity because of our assumption of directed links. \ddegree{T, D} is different from rooted spanning tree, because we allow rounds in which $G(V, \E(t))$ does not have a ``root.'' 

In order to solve consensus, the corresponding \ddegree{T, D} requires $G_t$ to have a root node over a $T$-round interval. Hence, it is tempting to view such an interval in our model as one ``mega-round'' in the prior model \cite{Nowak_Averaging_ICALP15,Nowak_tight_bound_asymptotic_JACM21,Kyrill_dynamicConsensus_ITCS23}. However, the ``dynamics'' of how nodes change their states is different in our model. Generally speaking, our model captures a more fine-grained interaction between nodes that have different ``views.'' This is because in a $T$-round interval, some nodes may update frequently while the other nodes only update once. Such a behavior is not captured when using the ``mega-round'' formulation. This perspective will become more clear when we discuss our algorithms.

\subsection{Exact Consensus and Approximate Consensus}

We study both exact consensus and approximate consensus problems \cite{welch_book,lamport_agreement,DolevLPSW86,Lynch_DistributedAlgorithms1996}. The former task requires that the fault-free nodes agree on exactly the same output, whereas the latter one only requires that the fault-free nodes to agree on roughly the same output. Formally, we have

\begin{definition}
    Binary exact consensus algorithms need to satisfy the following three conditions:
    \begin{itemize}
        \item[(i)] \textit{Termination}: Each fault-free node outputs a value;
        
        \item[(ii)] \textit{Validity}: The output of each fault-free  node equals to some binary input given to a non-Byzantine node; and 
        
        \item[(iii)] \textit{Agreement}: Fault-free nodes have an identical output.
    \end{itemize}
\end{definition}

\begin{definition}
    Approximate consensus algorithms need to satisfy the following three conditions:
    
        \begin{itemize}
        \item[(i)] \textit{Termination}: Each fault-free node outputs a value; 
        
        \item[(ii)] \textit{Validity}: The outputs of all fault-free nodes are within the convex hull of the non-Byzantine inputs; and 
        
        \item[(iii)] \textit{$\epsilon$-agreement}: The outputs of all fault-free nodes are within $\epsilon$ of each other's output.
    \end{itemize}

\end{definition}

Following prior work \cite{Tseng_PODC12,Nowak_tight_bound_asymptotic_JACM21,Lynch_DistributedAlgorithms1996}, we assume that the range of initial inputs of all nodes is bounded. Without loss of generality, we can scale the inputs to $[0,1]$ as long as they are bounded by scaling $\epsilon$ down by the same factor. Formally, for approximate consensus algorithms, we assume that every $i\in\calH$ has input $x_i\in [0,1]$. 



\subsection{Challenges in Anonymous Dynamic Networks}

Prior fault-tolerant consensus algorithms do not work in our model. There are mainly three categories of algorithms in synchronous or asynchronous message-passing networks: (i) algorithms that assume reliable message delivery, e.g., 
\cite{DolevLPSW86,LocalWait,Tseng_PODC12,AA_Fekete_aoptimal}; (ii) algorithms that relax termination (namely asymptotic consensus), e.g., 
\cite{Nowak_Averaging_ICALP15,Nowak_tight_bound_asymptotic_JACM21,Sundaram_ACC,Sundaram_journal}; and (iii) algorithms that piggyback the entire history (namely the full-information model) \cite{Kyrill_dynamicConsensus_ITCS23}. There are also some algorithms \cite{Kuhn_dynamic_SIGACTNews11,Kunh_dynamicConsensus_PODC11} that do not tolerate node faults. 

The analysis and design of these prior algorithms do not directly apply to our model because of three main challenges: (i) we cannot implement some well-known primitives such as reliable broadcast \cite{OR_Dolev_2004,bracha1987asynchronous}, because nodes are anonymous, and $T$ and $D$ are unknown; (ii) nodes are not able to piggyback extra information because of the anonymity and limited bandwidth assumption; and (iii) due to anonymity, it is not easy to break a tie.  


\section{Related Work}
\label{s:related}

Distributed tasks with fault-tolerance (of both node and link faults) have been widely studied \cite{welch_book,Lynch_DistributedAlgorithms1996}. We focus on the closely related work on dynamic networks in this section. Early works \cite{Gafni_Dynamic_FOCS87,Dynamic_STOC92,Welch_DynamicLeader_IPDPS09,Welch_Mutex_Adhoc01} focus on networks that eventually stop changing. These algorithms usually guarantee progress or liveness only after the network stabilizes (i.e., when network stops changing). 

Subsequently, various works investigate networks with continual dynamic changes. Kuhn et al. propose the idea of $T$-interval connectivity and study problems on election, counting, consensus, token dissemination, and clock synchronization \cite{Kuhn_DynamicClock_SPAA09,Kuhn_dynamic_STOC10,Kuhn_dynamic_SIGACTNews11,Kunh_dynamicConsensus_PODC11}. The links in their model are assumed to be bi-directional and nodes are assumed to be correct. Bolomi et al. \cite{Bonomi_DynamicRB_SSS18} investigate reliable broadcast primitive in dynamic networks with locally bounded Byzantine adversary. The nodes are assumed to have an identity. 

Di Luna et al. have a series of works on anonymous dynamic networks which investigate problems like counting, leader election, and arbitrary function computation \cite{Bonomi_DynamicAnonymousCounting_ICDCN14,Bonomi_DynamicAnonymousCounting_ICDCS14,DiLuna_DynamicAnonymous_DISC23,DiLuna_DynamicAnonymous_PODC23,DiLuna_DynamicAnonymous_PODC15,DiLuna_DynamicAnonymous_OPODIS15}. There is also a line of works on investigating message adversary for consensus in directed dynamic graphs \cite{Kyrill_DynamicConsensus_DC19,Kyrill_DynamicConsensus_TCS18,Kyrill_dynamicConsensus_ITCS23,Nowak_Averaging_ICALP15,Nowak_tight_bound_asymptotic_JACM21}. 
These work do not assume node faults and link bandwidth is assumed to be unlimited. Therefore, their techniques are quite different from ours.

\begin{algorithm*}[t]
\caption{DAC: Steps at each node $i$ in round $t$. Node $i$ outputs $v_i$ when $p_i = p_{end}$ (identified in (\ref{eq:p-end-crash}))}
\label{alg:DAC}
\begin{algorithmic}[1]
\footnotesize
\item[{\bf Initialization:}]{}
    
    \item[] $v_i, v_{\min,i}, v_{\max,i}\gets x_i$ \Comment{$x_i$ is the input}
    
    \item[] $p_i\gets 0$ \Comment{phase index}
    
    \item[] $R_i\gets$ zero vector of length $n$
    
    \item[] $R_i[i] \gets 1$

    \item[] \hrulefill
    \vspace{-15pt}
    \begin{multicols}{2}    
    
    \For{$t \gets 0$ to $\infty$}
        \State \textbf{broadcast} $\langle i,v_i, p_i \rangle$ to all 
        \State $M_i \gets$ messages received  in round $r$
        \For{each $\langle j,v_j, p_j \rangle$ received from \textbf{port} $j$ in $M_i$}\label{line:ACR-start}
            \If {$p_j > p_i$} \label{line:ACR-jump1}
                \State $v_i\gets v_j$
                \State $p_i\gets p_j$ \label{line:ACR-jump}
                \State \textsc{Reset}() \label{line:ACR-jump2}
            \ElsIf {$p_j = p_i$ and $R_i[j]=0$}\label{line:ACR-phase-p}
                \State $R_i[j]\gets 1$
                \State $\textsc{Store}(v_j)$ \label{line:ACR-add-state}
            
                \If{$|R_i| \geq \lfloor \frac{n}{2} \rfloor + 1$} 
                \label{line:ACR-avg1}
                    \State $v_i\gets \frac{1}{2}(v_{\min,i}+v_{\max,i})$
                    \State $p_i\gets p_i+1$
                    \State \textsc{Reset}() \label{line:ACR-end}
                \EndIf
                
                \If{$p_i = p_{end}$}\Comment{$p_{end}$ in  (\ref{eq:p-end-crash})}
                    \State output $v_i$
                \EndIf
            \EndIf 
        \EndFor
    \EndFor    

    \vspace{5pt}
    \Function{Reset}{\null}
        \State $R_i[i]\gets 1$; ~~and $R_i[j]\gets 0,\forall j\neq i$
        \State $v_{\min,i}, v_{\max,i}\gets v_i$ 
    \EndFunction
    
    \vspace{5pt}
    \Function{Store}{$v_j$}
        \If {$v_j < v_{\min,i}$}
            \State $v_{\min,i} \gets v_j$
        \ElsIf {$v_j > v_{\max,i}$}
            \State $v_{\max,i} \gets v_j$
        \EndIf
    \EndFunction
\end{multicols}   
\vspace{-15pt}
\end{algorithmic}
\end{algorithm*}

\section{Crash-Tolerant Approximate Consensus}
\label{sec:crash}

This section considers the case when nodes may crash. 
We present a simple algorithm, \textit{DAC}, which stands for Dynamic Approximate Consensus and achieves approximate consensus if $n\geq 2f+1$ and the graph satisfies \ddegree{T, \lfloor n/2 \rfloor}. Section \ref{s:impossible} will prove that these conditions together are necessary. DAC shows that they are sufficient. 

DAC is inspired by prior algorithms \cite{DolevLPSW86,Tseng_PODC12,Lynch_DistributedAlgorithms1996} in which nodes proceed in phases and nodes collect messages from a certain phase to update its local state value and proceed to the next phase. DAC has two key changes: 

\begin{itemize}
    \item The capability of ``jumping'' to a future phase when receiving a state with a larger phase index; and

    \item Each node uses a bit vector $R_i$ to keep track of received messages from the same phase. This is possible because of our assumption of port numbering (cf. Section \ref{s:preliminaries}). 
\end{itemize}
Change (i) avoids sending repetitive state values (state values from prior phases) to handle message loss under limited bandwidth, while preserving the same convergence rate. Change (ii) allows node $i$ to know when it is safe to proceed to the next phase. DAC is presented in Algorithm \ref{alg:DAC}. 

Each node $i$ initializes its local state $v_i$ to a given input $x_i$, and then broadcasts its state $v_i$ in every round. It only stores two states -- $v_{min, i}$ and $v_{max, i}$ -- the smallest and largest phase-$p$ states observed so far, respectively. Each node $i$ in phase $p_i$ has two ways to proceed to a higher phase: 
\begin{enumerate}
    \item (line \ref{line:ACR-jump1} -- \ref{line:ACR-jump2}): upon receiving a message from phase $q > p_i$, node $i$ directly copies the received state and ``jumps'' to phase $q$; or
    
    \item (line \ref{line:ACR-avg1} -- \ref{line:ACR-end}): upon receiving $ \lfloor \frac{n}{2} \rfloor + 1$ phase-$p_i$ states from different nodes, node $i$ updates its state and proceeds to phase $p+1$. 
\end{enumerate}
Since a node might not receive enough phase-$p_i$ messages in a round, it uses an $n$-bit bit vector to keep track of the senders. Recall that even though our model does \textit{not} assume node identity, each node can still use \underline{receiving ports} to distinguish messages received. For brevity, denote the number of ones in vector $R_i$ by $|R_i|$.\footnote{One practical optimization is to use an $(n-1)$-bit bit vector. In our design, $R_i[i]$ is always $1$, since  each node always has its own state value.} Every fault-free node repeats this process until it proceeds to the termination phase $p_{end}$, which will be defined later in Equation (\ref{eq:p-end-crash}). 

\subsection{Correctness of DAC}
\noindent\textbf{Technical Challenge.}~~Intuitively, the algorithm is correct, because a ``future'' state value is better, in terms of convergence, than the current state value (i.e., the $v_i$'s from the current phase $p$). In other words, ``jump'' should not affect correctness. However, because nodes do not use state values from the same phase to update their state values, we cannot directly apply prior proofs (e.g., \cite{Lynch_DistributedAlgorithms1996,LocalWait,OR_Dolev_2004,Tseng_PODC12,DolevLPSW86}), which rely on the fact that a pair of nodes receive at least \textit{one common value} for each phase (via a typical quorum intersection argument). 

In DAC, nodes may move or jump to a phase, because they do not use the same updating rule. The key challenge is to devise the right setup so that we can use induction to derive a bound on the convergence rate. Our induction-based proof is useful for handling the case when nodes may not receive common values (due to message loss). 
Intuitively, we need to find a way to replace a common value by a common interval (that contain state values across different phases) for proving convergence. 

Algorithm DAC satisfies termination because after $T$ rounds, each node must receive either $\lfloor n/2 \rfloor + 1$ messages in the same phase (including message received from itself) or a message with a higher phase owing to \ddegree{T, \lfloor n/2 \rfloor}. Validity is also straightforward because of our updating rules and the assumption of non-Byzantine behavior. Hence, we focus on the $\epsilon$-agreement below.

\vspace{3pt}
\noindent\textbf{Notations.}~~
We introduce two useful notations.

\begin{definition}
    Let $S$ be a finite multiset. Define the \textit{cardinality} $|S|$ as number of elements in $S$ counting multiplicity, the \textit{range} of $S$ as $\range(S)=\max(S)-\min(S)$, and the \textit{interval} of $S$ as $\interval(S)=[\min(S),\max(S)]$. 
\end{definition}

\begin{definition}
    \label{def:ACR-Vp}
    Define $V^{(p)}$ as a multiset of phase-$p$ states of all nodes that have not crashed yet. 
\end{definition}
For a faulty node which crashed before phase $p$, its phase-$p$ state is empty and hence is excluded in $V^{(p)}$. Due to the ``jump'' feature of DAC, some fault-free nodes may skip a particular phase $p$. To simplify our analysis, we introduce the following definition:

\begin{definition}
    \label{def:ACR-vp-skip}
    If a node $i$ jumps from some phase $p$ to phase $q>p$, then we define its phase-$p'$ state value $v_i^{p'}$ of skipped phases ($p < p' <q$) as $v_i^q$.
\end{definition}
Denote $n_p=|V^{(p)}|$. Definition \ref{def:ACR-vp-skip} and the assumption that $n \ge 2f+1$  imply that $n_p \geq n-f \ge \lfloor\frac{n}{2}\rfloor + 1$ for all phase $p$. 
Without loss of generality, we order $V^{(p)}$ \textit{chronologically} such that the skipped state values are ordered last (breaking tie arbitrarily). In other words, $V^{(p)}=\{v_1^p,\ldots,v_{n_p}^p\}$, where $v_k^p$ is the phase-$p$ state of the $k$-th node that proceeded to phase $p$. For nodes that skip state $p$, their values appear last in $V^{(p)}$. 

With $V^{(p)}$ defined, we can introduce the notion of convergence rate.

\begin{definition}[Convergence Rate]
\label{def:convergence-rate}
    Consider an algorithm $\calA$ in which each node $i$ maintains a state value $v_i$. 
    Then we say $\calA$ has convergence rate $\rho$, for some $\rho\in[0,1]$, if 
    $\range(V^{(p+1)}) \le \rho\cdot \range(V^{(p)})$.
    \qinzi{}
\end{definition}

Finally, we need two more definitions to help our proof. Define $V^{(p)}_k=\{v_1^p,\ldots,v_k^p\}$, i.e., the first $k$ elements in $V^{(p)}$. 
Define $W^{(p)}$ as $V^{(p)}$ sorted in ascending order of values, i.e., $W^{(p)}=\{w_1^p,\ldots, w_{n_p}^p\}$ such that $w_1^p\leq \ldots \leq w_{n_p}^p$. Note that $W^{(p)}$ is defined with respect to the values, instead of chronological order. 

\vspace{3pt}
\noindent\textbf{Convergence Proof.}~~For $\epsilon$-agreement, we first prove the following key lemma to identify the convergence rate. Roughly speaking, the convergence rate identifies the ratio that the range of fault-free nodes decreases in each phase. The base case is similar to the ``common value'' technique in \cite{LocalWait,AA_Dolev_1986,OR_Dolev_2004}. The difference lies in the inductive step where we need to consider nodes that skip phases due to the ``jump'' updating rule. Later in Section \ref{sec:byz}, we generalize the technique to handle Byzantine nodes, where the proof naturally becomes more complicated.

\begin{lemma}
For each $p~~(0 \leq p \leq p_{end}-1)$ and $k\in[n_{p+1}]$,
\begin{equation}
    \label{eq:SmallAC-lem}
    V_k^{(p+1)} \subseteq \left[\frac{w_1^p+w_{\lfloor\frac{n}{2}\rfloor+1}^p}{2}, \frac{w_{n_p-\lfloor\frac{n}{2}\rfloor}^p+w_{n_p}^p}{2} \right].
\end{equation}
\end{lemma}

\begin{proof}
We prove the lemma by induction on $k$. First, we prove an important claim.

\begin{algorithm*}[t]
\caption{DBAC: Steps at each node $i$ in round $t$. Node $i$ outputs $v_i$ when $p_i = p_{end}$ (identified in (\ref{eq:p-end-byz}))}
\label{alg:BACR}
\begin{algorithmic}[1]
\footnotesize
\item[{\bf Initialization:}]{}
    
    \item[] $v_i \gets x_i$ \Comment{$x_i$ is the input}
    
    \item[] $p_i\gets 0$ \Comment{phase index}
    
    \item[] $R_i\gets$ zero vector of length $n$
    
    \item[] $R_i[i] \gets 1$ 
    
    \item[] $R_{i,\low}, R_{i,\high}\gets \{\}$
    
    \item[] \hrulefill
    \vspace{-15pt}
    \begin{multicols}{2} 
    \For{$t \gets 0$ to $\infty$}
    \State broadcast $\langle i,v_i,p_i \rangle$ to all
    
    \State $M_i \gets$ messages received in round $r$
    
    \For{each $\langle j,v_j,p_j\rangle$ from port $j$ in $M_i$} \label{line:BACR-atomic1}
        \If {$p_j \geq p_i$ and $R_i[j]=0$}
            \State $R_i[j]\gets 1$
            \State $\textsc{Store}(v_j)$
        \EndIf 
        \If {$|R_i|\geq  \lfloor\frac{n+3f}{2}\rfloor + 1$} \label{line:BACR-receive}
            \State $v_i\gets \frac{1}{2}(\max(R_{i,\low})+\min(R_{i,\high}))$
            \State $p_i\gets p_i+1$
            \State \textsc{Reset}()
            \label{line:BACR-atomic2}
        \EndIf
    \EndFor
    
        \If{$p_i = p_{end}$}
            \State output $v_i$
        \EndIf
    \EndFor

    \vspace{5pt}
    \Function{Reset}{\null}
        \State 
        $R_i[j]\gets0, \forall j\neq i$ 
        \State $R_{i,\low},R_{i,\high}\gets \{\}$
    \EndFunction
    
    \vspace{5pt}
    \Function{Store}{$v_j$}
        \If {$|R_{i,\low}|\leq f+1$}
            \State $R_{i,\low}\gets R_{i,\low}\cup \{v_j\}$
        \ElsIf {$v_j < \max(R_{i,\low})$}
            \State replace max value in $R_{i,\low}$ with $v_j$
        \EndIf
        
        \If {$|R_{i,\high}|\leq f+1$}
            \State $R_{i,\high}\gets R_{i,\high}\cup \{v_j\}$
        \ElsIf {$v_j>\min(R_{i,\high})$}
            \State replace min value in $R_{i,\high}$ with $v_j$
        \EndIf 
    \EndFunction
\end{multicols}
\vspace{-15pt}
\end{algorithmic}
\end{algorithm*}

\begin{claim}
\label{claim:AC-trim}
For every $p\geq 0$, if a node in phase $p$ updates to phase $p+1$ by receiving $\lfloor n/2 \rfloor + 1$ phase-$p$ states, then its new state value $v$ in phase $p+1$ satisfies
\begin{equation*}
    v\in \left[\frac{w_1^p+w_{\lfloor n/2 \rfloor+1}^p}{2}, \frac{w_{n_p-\lfloor n/2 \rfloor }^p+w_{n_p}^p}{2} \right].
\end{equation*}
\end{claim}

\begin{proof}[Proof of Claim \ref{claim:AC-trim}]
In each phase, the state value $v_i$ of each node remains unchanged until the node updates to the next phase. 
Moreover, line \ref{line:ACR-phase-p} ensures that each state is received at most once by a receiver in each phase. Therefore, if some node in phase $p$ receives $\lfloor n/2 \rfloor + 1$ phase-$p$ states (including from itself), then the smallest $\lfloor n/2 \rfloor + 1$ possible states it can receive are $w_1^p,\ldots, w_{\lfloor n/2 \rfloor + 1}^p$. Similarly, the maximum possible states are $w_{n_p}^p,\ldots,w_{n_p-\lfloor n/2 \rfloor}^p$. In conclusion, the new state in phase $p+1$ falls in the interval in Claim \ref{claim:AC-trim}.
\end{proof}

Claim \ref{claim:AC-trim} proves the base case that Equation (\ref{eq:SmallAC-lem}) holds for fixed $k=1$ and for every $p\geq 0$ because $V_1^{(p+1)}$ only consists of the state value of one node, which must update to phase $p+1$ by receiving $\lfloor n/2 \rfloor + 1$ phase-$p$ states.

In the induction case, assume Equation (\ref{eq:SmallAC-lem}) holds for every $p\geq 0$ and for some $k$, and we want to prove Equation (\ref{eq:SmallAC-lem}) for every $p\geq 0$ and for $k+1$. Note that the new node proceeds to phase $p+1$ by either receiving $\lfloor n/2 \rfloor + 1$ phase-$p$ state values (including one from itself) or copying a future state. In the former case, Claim \ref{claim:AC-trim} implies the induction statement, whereas in the latter case, the range of $V_{k+1}^{(p+1)}$ is unchanged and therefore Equation (\ref{eq:SmallAC-lem}) again holds for $k+1$.
\end{proof}

\begin{remark}
By definition, $n \geq n_p$. This implies that $n - \lfloor n/2 \rfloor \geq n_p - \lfloor n/2 \rfloor$, which leads to $\lfloor n/2 \rfloor+1 \geq n_p - \lfloor n/2 \rfloor$. Since $w^p$'s are ordered in the  ascending order, we have $ w_{\lfloor n/2 \rfloor + 1}^p\geq w_{{n_p}-\lfloor n/2 \rfloor}^p$, and thus
\begin{align*}
    \range(V^{(p+1)})
    &\leq \frac{w_{n_p-\lfloor n/2 \rfloor}^p+w_{n_p}^p}{2} - \frac{w_1^p+w_{\lfloor n/2 \rfloor + 1}^p}{2} \\
    &\leq \frac{w_{n_p}^p-w_1^p}{2} = \frac{1}{2}\cdot \range(V^{(p)}).
\end{align*}
In other words, Algorithm \ref{alg:DAC} converges with rate $\frac{1}{2}$. 
\end{remark}

This implies the following theorem, which identifies the phase $p_{end}$ in which node $i$ is able to output $v_i$. 

\begin{theorem}
\label{thm:DAC-converge}
    Algorithm DAC satisfies $\epsilon$-agreement after phase $p_{\mathrm{end}}$, where
    \begin{equation}
        \label{eq:p-end-crash}
        p_{\mathrm{end}} = \log_{\frac{1}{2}}(\epsilon) 
    \end{equation}
\end{theorem}

Interestingly, the lower bound from \cite{Nowak_tight_bound_asymptotic_JACM21} shows that $1/2$ is the optimal rate for any fault-tolerant approximate consensus algorithms, even in static graphs. Hence, DAC achieves the optimal convergence rate and optimal resilience even in the static graph with only node crash faults.

\section{Byzantine Approximate Consensus}
\label{sec:byz}

This section considers up to $f$ Byzantine nodes (denoted as set $\calB$), and the rest of the nodes (denoted as set $\calH$) are fault-free and always follow the algorithm specification. We present an approximate consensus algorithm, \textit{DBAC}, which stands for Dynamic Byzantine Approximate Consensus and is correct if $n\geq 5f+1$ and $G(V,E)$ satisfies \ddegree{T, \lfloor (n+3f)/2 \rfloor}. 

Algorithm DBAC and Algorithm DAC share a similar structure, but DBAC has different update rules to cope with Byzantine faults. Plus, nodes do \textit{not} skip phases in DBAC. The pseudo-code is presented in Algorithm \ref{alg:BACR}.

Each node starts with phase $0$ and initializes its local state value $v_i$ to the given input $x_i$. Then it broadcasts its current local state value in every round. For each node in phase $p_i$, upon receiving $\lfloor\frac{n+3f}{2}\rfloor + 1$ state values \underline{from phase $p_i$ or higher}, it updates its local state value $v_i$ to the average of the ($f+1$)-st lowest state value and the ($f+1$)-st highest state value that have been received so far and then proceeds to phase $p+1$. To achieve the goal, node $i$ uses $R_{i,\low}$ and $R_{i,\high}$ -- lists that store the $f+1$ lowest and $f+1$ highest received states in phase $p$ or higher, respectively. Recall that $|R_i|$ denotes the number of ones in $R_i$, whereas $|R_{i,\low}|$ and $|R_{i,\high}|$ denote the cardinality (i.e., the number of elements) of $R_{i,\low}$ and $R_{i,\high}$, respectively.

Our update rule ensures that the new state value falls in the range of fault-free state values despite of the existence of Byzantine messages. Moreover, since at most $f$ nodes are Byzantine faulty and the graph is assumed to satisfy \ddegree{T, \lfloor (n+3f)/2 \rfloor}, this step is always non-blocking. (Recall that a node can receive a message from itself as well.) Every node repeats this process until phase $p_{end}$, whose value will be determined later in Equation (\ref{eq:p-end-byz}).

DBAC is inspired by the iterative Byzantine approximate consensus algorithm (BAC) \cite{AA_Dolev_1986}, which update states using states from the same phase. BAC relies on reliable channels; hence, is not feasible in our dynamic network model. DBAC can update states using messages from different phases (as shown in the $\textsc{Store}(-)$ function below). These differences allow us to tolerate the nature of dynamic network; however, using messages from different phases make the correctness proof more complicated than prior analysis.

There exists a Byzantine approximate consensus algorithm \cite{OR_Dolev_2004} that achieves an optimal resilience $n \geq 3f+1$ in a static graph with only Byzantine nodes; however, it uses a stronger primitive, reliable broadcast \cite{bracha1987asynchronous}, and a technique of witness (of certain state values). Because of the anonymity assumption, such techniques are not possible in our model.

\subsection{Correctness of DBAC}

\begin{theorem}
\label{thm:smallbac-termination}
    Algorithm DBAC satisfies termination.
\end{theorem}

\begin{proof}
    We prove termination by induction on phase $p\geq 0$. Formally, we define the induction statement as: every fault-free node proceeds to phase $p$ for $1 \leq p \leq p_{end}$ within finite number of rounds. The base case holds because all nodes are initially in phase $0$. Now suppose all fault-free nodes proceed to phase $p$ within a finite number of rounds. Then after all fault-free nodes are in phase $p$ or higher, by assumption of \ddegree{T, \lfloor (n+3f)/2 \rfloor}, every fault-free node receive at least $\lfloor (n+3f)/2 \rfloor+1$ state values from fault-free nodes in phase $p$ or higher within $T$ rounds (including one from itself). Hence, every fault-free node proceeds to the next phase according to line \ref{line:BACR-receive} -- \ref{line:BACR-atomic2}, which proves the induction.    
\end{proof}

We define $V^{(p)}$ and $W^{(p)}$ in a similar way as we did in the previous section. The difference is that we only consider ``\textit{fault-free} nodes.'' Recall that in the case of crash faults, $V^{(p)}$ and $W^{(p)}$ might include nodes that crash later in phases after phase $p$. In the Byzantine case, we exclude any Byzantine state values, as the values are not well-defined. 

We then sort $V^{(p)}=\{v_1^p,\ldots,v_{|V^{(p)}|}^p\}$ chronologically, i.e., in the increasing order of round index in which the state value is calculated (breaking ties arbitrarily). Since we have already proved that DBAC terminates, $|V^{(p)}|=h$ for all $p\geq0$, where $h$ is the number of fault-free  nodes and $h=|\calH|\geq n-f$. For easiness of calculation, we also introduce the following notations:



\begin{definition}
Define $W^{(p)}=\{w_1^p,\ldots,w_h^p\}$ as $V^{(p)}$ ordered by values, i.e., $w_1^p\leq \ldots \leq w_h^p$.
\end{definition}


\begin{definition}
Define $U=\{u_1,\ldots,u_b\}$ as a multiset of \textit{arbitrary} values from Byzantine nodes, where $b$ is the number of Byzantine nodes in the execution and $b=|\calB|\leq f$.
\end{definition}


\begin{definition}
\label{def:k_t}
For round $t$ and phase $p$ ($0\le p\le p_{\mathrm{end}}$), define $k(t,p)$ as the number of fault-free nodes that are in phase $p$ or higher at the start of round $t$. 

Moreover, define $V_t^{(p)} = \{v_1^p,\ldots,v_{k(t,p)}^p\}$, i.e., the first $k(t,p)$ elements in the multiset $V^{(p)}$. If $k(t,p)=0$, we define $V_t^{(p)}=\emptyset$. In other words, $V_t^{(p)}$ is the multiset of phase-$p$ states of fault-free nodes whose phases are $\ge p$ at the start of round $t$.
\end{definition}

\begin{remark}
\label{rmk:bac}
Observe the properties below for $k(t,p)$:
\begin{enumerate}
    \item For fixed $p\geq 0$, $k(t,p)$ is non-decreasing with respect to $t$, i.e., $t \le t'$ implies 
    
    ~~~~~~~~~~~~$k(t,p)\le k(t',p)$ and thus $V_t^{(p)} \subseteq V_{t'}^{(p)}$.
    
    \item For fixed $t\ge 0$, $k(t,p)$ is non-increasing with respect to $p$, i.e., $p\le q$ implies 
    
    ~~~~~~~~~~~~$k(t,p) \geq k(t,q)$.\footnote{Although it still holds that $V_t^{(p)} \supseteq V_t^{(q)}$, the proof is not immediate. This identity turns out to be the key to the proof of Lemma \ref{lem:BACR-weak-validity}.}

    \item When $t=0$, $k(0,0)=h$ and $k(0,p)=0$ for all $p>0$ because all nodes are initially in phase $0$. Consequently, $V_0^{(0)} = V^{(0)}$ and $V_0^{(p)}=\emptyset$ for $p>0$.
    
    \item By termination, every fault-free node updates to phase $p_{\mathrm{end}}$ within finite time. Therefore, there exists a finite $t_{\mathrm{end}}$ that is the last round in which a fault-free node updates to $p_{\mathrm{end}}$. Moreover, $k(t_{\mathrm{end}}, p)=h$ and $V_{t_{\mathrm{end}}}^{(p)} = V^{(p)}$ for all $p$.
\end{enumerate}
\end{remark}

We are now ready to prove the key lemma that bounds the range of fault-free values. Recall that $\interval(V)=[\min(V),\max(V)]$ and $\range(V) = |\max(V)-\min(V)|$.


\begin{lemma} 
    \label{lem:BACR-weak-validity}
    For every round $t \ge 0$,
    \begin{equation}
    \label{eq:BAC-lem-validity}
       \interval(V_t^{(q)}) \subseteq \interval (V_t^{(p)}) ,\ \forall 0 \leq p \le q.
    \end{equation}
\end{lemma}

In other words, Lemma \ref{lem:BACR-weak-validity} suggests that in every round, higher-phase states are within lower-phase states.

\begin{proof}
We prove the lemma by induction on $t$. 

In the base case when $t=0$, recall Remark \ref{rmk:bac} that $V_0^{(0)}=V^{(0)}$ and $V_0^{(p)}=\emptyset$ for all $p>0$, and so \eqref{eq:BAC-lem-validity} trivially holds.
    
For the induction case, assume \eqref{eq:BAC-lem-validity} holds for rounds $t$, and we want to prove for $t+1$. The key is to show that
\begin{equation}
    \interval(V_{t+1}^{(p+1)}) \subseteq \interval( V_{t}^{(p)} ),\ \forall \, p\ge 0.
    \label{eq:claim-bac-validity}
\end{equation}
Recall that $V_t^{(p)} \subseteq V_{t+1}^{(p)}$ by Remark \ref{rmk:bac}. Together with \eqref{eq:claim-bac-validity}, we have $\interval(V_{t+1}^{(p+1)}) \subseteq \interval(V_{t+1}^{(p)})$, which proves the induction case. The rest of the proof aims to prove \eqref{eq:claim-bac-validity}.

If no node updates from phase $p$ to $p+1$ in round $t$, then $V_{t+1}^{(p+1)} = V_{t}^{(p+1)}$ and \eqref{eq:claim-bac-validity} follows from the induction assumption.
Otherwise, consider some node $i$ that updates from phase $p$ to $p+1$ in round $t$. Its new state in $V_{t+1}^{(p+1)}$ is of form $v=\frac{1}{2}(\max (R_{i,\mathrm{low}})+ \min (R_{i,\mathrm{high}}))$. For both $R_{i,\mathrm{low}}$ and $R_{i,\mathrm{high}}$, they consist of $f+1$ messages each of which either comes from a Byzantine node or is in $V_t^{(q)}$ for some $q\ge p$. Since there are at most $f$ Byzantine nodes, there exist $q,q'\ge p$ and $u\in V_t^{(q)}, w\in V_t^{(q')}$ such that
\begin{equation*}
    u \le \max (R_{i,\mathrm{low}}) \le \min (R_{i,\mathrm{high}}) \le w.
\end{equation*}
By induction assumption, $V_t^{(q)}, V_t^{(q')} \subseteq V_t^{(p)}$. Consequently, $u \le v \le w$ and $v\in \interval(V_t^{(p)})$. This proves \eqref{eq:claim-bac-validity}. 
\end{proof}

Lemma \ref{lem:BACR-weak-validity} implies the validity of Algorithm DBAC upon substituting $t=t_{\mathrm{end}}$, $p=0$ and $q=p_{\pend}$ into Equation \eqref{eq:BAC-lem-validity}, which together with Remark \ref{rmk:bac} implies that
\begin{equation*}
    V^{(p_{\pend})} \subseteq \interval(V^{(0)}).
\end{equation*}


\vspace{3pt}
\noindent\textbf{DBAC: $\epsilon$-agreement}.~~We next present the proof for convergence. 
In addition to the effect of Byzantine values, we also need to consider the case when a node uses values from different phases when updating. This is more complicated to analyze than the case of DAC, since in our prior analysis, a node simply jumps to a future state that trivially satisfies the induction statement. Also, we cannot apply prior proofs for the traditional Byzantine fault model with  reliable channel (e.g., \cite{AA_Dolev_1986,LocalWait,AA_nancy,OR_Dolev_2004}) either. This is again because a pair of fault-free nodes may not use a common value to update their future states. A technical contribution is to identify a setup to use induction to prove the convergence. 

We need to prove that the new state value at a fault-free node falls in a smaller interval as it updates to a higher phase. In the base case, each node in phase $p$ receives non-Byzantine messages from a fixed multiset $V^{(p)}$. By the classical common value analysis and quorum intersection argument  \cite{AA_nancy,AA_Dolev_1986}, all nodes in the base case must receive at least one common value from a fault-free node. 

In the more general inductive step, this technique no longer works, because each node can also receive messages from higher phase(s). We need to show that all fault-free nodes must share some \underline{common information}. Even though each pair of fault-free nodes may \textit{not} receive a common value, we show that each fault-free node must receive at least one non-Byzantine message in a ``\textit{common multiset}'' (a generalized concept of common value). We then bound the range of this common multiset using $a_k^p$ and $A_k^p$, defined below. The common multiset allows us to derive the desired convergence rate.

\begin{definition}
For each $p$, define $a_k^p$ and $A_k^p$ recursively as:
\[
a_{k+1}^p = (a_k^p+w_{1}^p)/2, \enspace A_{k+1}^p = (A_k^p+w_{h}^p)/2,
\]
with initial values $a_0^p=A_0^p=w_{2f+1}^p$.
\end{definition}

Notice two useful properties. First, since $w_1^p\leq w_{2f+1}^p\leq w_h^p$, we have
\[
w_1^p\leq a_k^p \leq w_{2f+1}^p\leq A_k^p \leq w_h^p.
\]
Moreover, using geometric series, their explicit formulas are:
\begin{align*}
a_k^p 
&= 2^{-k}w_{2f+1}^p+\sum_{i=1}^k 2^{-i}w_1^p 
= w_1^p+2^{-k}(w_{2f+1}^p-w_1^p).
\end{align*}
Similarly, $A_k^p=w_h^p+2^{-k}(w_{2f+1}^p-w_h^p)$.

We are now ready to present the full proof below, and the illustration of common multiset is presented in Figure \ref{fig:BACR-agreement-sync}. 
Recall that $v_k^{p+1}$ denotes the phase-$(p+1)$ state of the $k$-th fault-free node that updates to phase $p+1$.

\begin{lemma}
\label{lem:BACR-agreement}
Suppose $n\geq 5f+1$. Then for every $k\in [h]$,
\begin{equation}
    \label{eq:BAC-thm-agreement}
    v_k^{p+1}\in[a_k^p, A_k^p],\ \forall p\geq 0.
\end{equation}
\end{lemma}

\begin{proof}
    We prove the theorem by induction on $k$. 
    
    
    \vspace{3pt}
    \noindent\textbf{Base Case:}~
    In base case, fix $k=1$ and consider phase $p\geq 0$. Assume that node $i^*$ is the first fault-free node that proceeds to phase $p+1$ (breaking ties arbitrarily), and denote by $R_{\low},R_{\high}$ the recording lists of node $i^*$ in phase $p$.
    
    By construction, all fault-free states received by $i^*$ must be in phase $p$. In addition, at most $f$ received state values are Byzantine. Therefore, $\max(R_{\low})\geq w_1^p$ because $\max(R_{\low})$ reaches its minimum possible value in the worst case when $i^*$ receives these following $\lfloor (n+3f)/2 \rfloor+1$ states:
    \begin{equation*}
        u_1 \leq \ldots \leq u_f \leq w_1^p \leq \ldots \leq w_{\lfloor (n+3f)/2 \rfloor+1-f}^p.
    \end{equation*}
    Similarly, $\min(R_{\high})\geq w_{\lfloor (n+3f)/2 \rfloor+1-2f}^p$. Also, since $n\geq 5f+1$, $w_{\lfloor (n+3f)/2 \rfloor+1-2f}^p \geq w_{2f+1}^p$. Therefore,
    \[
    v_1^{p+1} = \frac{\max(R_{\low})+\min(R_{\high})}{2} \geq \frac{w_1^p+w_{2f+1}^p}{2} = a_1^p.
    \]
    
    Symmetrically, $\max(R_{\low}) \leq w_{2f+1}^p$ and $\min(R_{\high})\leq w_{\lfloor (n+3f)/2 \rfloor+1-2f}^p$. Recall that $h=|\calH|\geq n-f$ and $n \geq 5f+1$, so $w_{\lfloor (n+3f)/2 \rfloor+1-2f}^p\leq w_h^p$. Thus,
    \[
    v_1^{p+1} \leq \frac{w_{2f+1}^p+w_h^p}{2} = A_1^p.
    \]
    In conclusion, $v_1^{p+1}\in[a_1^p,A_1^p]$, for all $p\geq 0$.
    
    
    \vspace{3pt}
    \noindent\textbf{Induction Case:}~
    In induction case, assume Equation (\ref{eq:BAC-thm-agreement}) is true for all $i\in [k]$, and we want to prove for $k+1$. Consider an arbitrary phase $p$ and assume that node $j^*$ is the $(k+1)$-st fault-free node that proceeds to phase $p+1$ in round $t$. Denote $R_{\low},R_{\high}$ as recording lists of node $j^*$ in phase $p$, and denote the received state values of node $j^*$ as $r_1\leq\ldots\leq r_{\lfloor (n+3f)/2 \rfloor+1}$. Then $\max(R_{\low})=r_{f+1}$ and $\min(R_{\high})=r_{\lfloor (n+3f)/2 \rfloor+1-f}$.
    
    Note that every received state value $r$ must come from one of the three possible sources: (i) a Byzantine node; (ii) $r\in V_t^{(p)}$; or (iii) $r\in V_t^{(q)}$ for some $q>p$. Let's consider the latter two cases. 
    
\vspace{3pt}
\noindent\textit{First case}: Suppose $r\in V_t^{(q)}$ for some $q>p$. Then
\begin{align*}
    r \in \interval(V_t^{(q)}) 
    &\stackrel{(i)}{\subseteq} \interval(V_t^{(p+1)}) 
    \stackrel{(ii)}{\subseteq} [a_k^p, A_k^p].
\end{align*}
Here (i) follows from Lemma \ref{lem:BACR-weak-validity} and (ii) follows from the induction assumption.

\vspace{3pt}
\noindent\textit{Second case}: Suppose $r\in V_k^{(p)}$. Recall that $a_k^p\leq w_{2f+1}^p\leq A_k^p$, so we can partition $\interval(V^{(p)})$ into three parts as in Figure \ref{fig:BACR-agreement-sync}: $V_1=[w_1^p,a_k^p)$, $V_2=[a_k^p,A_k^p]$, and $V_3=(A_k^p,w_h^p]$.
    
    \begin{figure}[H]
        \centering
        \begin{tikzpicture}
        \draw[thick] (0,0) -- (7,0);
        
        \draw[thick] (0,-0.1) -- (0,0.1);
        \node (b) at (0,-0.5) {$w_1^p$};
        
        \draw[thick] (7,-0.1) -- (7,0.1);
        \node (c) at (7,-0.5) {$w_h^p$};
        
        \draw[thick] (1.5,-0.1) -- (1.5,0.1);
        \node (c) at (1.5,-0.5) {$a_k^p$};
        
        \draw[thick] (5.5,-0.1) -- (5.5,0.1);
        \node (c) at (5.5,-0.5) {$A_k^p$};

        \draw[thick] (3.5,-0.1) -- (3.5,0.1);
        \node (c) at (3.75,-0.5) {$w_{2f+1}^p$};
        
        \draw [thick,decorate,decoration={brace,amplitude=3pt}] (0,0.2) -- (1.5,0.2) node [black,midway,yshift=10pt] {$V_1$};
        \draw [thick,decorate,decoration={brace,amplitude=3pt}] (1.5,0.2) -- (5.5,0.2) node [black,midway,yshift=10pt] {$V_2$};
        \draw [thick,decorate,decoration={brace,amplitude=3pt}] (5.5,0.2) -- (7,0.2) node [black,midway,yshift=10pt] {$V_3$};
        
        \end{tikzpicture}
        \caption{Partition of $\interval(V^{(p)})$ into $V_1,V_2$ and $V_3$.}
        \label{fig:BACR-agreement-sync}
    \end{figure}
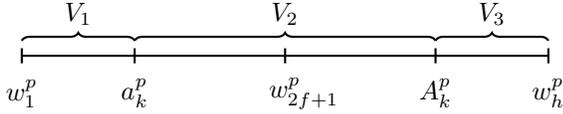
    
    By definition of $w_{2f+1}^p$, it is the $(2f+1)$-st largest state in $V^{(p)}$. Since $a_k^p\leq w_{2f+1}^p$, $r< a_k^p$ implies $r<w_{2f+1}^p$ and thus at most $2f$ state values from $V^{(p)}$ fall in $V_1$. 
    
    In conclusion, among the $\lfloor (n+3f)/2 \rfloor+1$ received state values, at most $f$ are Byzantine, and the rest are fault-free and greater than or equal to $w_1^p$. Moreover, among the fault-free state values, at most $2f$ are less than $a_k^p$. Hence, $r_{f+1}\geq w_1^p$ and $r_{3f+1}\geq a_k^p$. Finally, since $n\geq 5f+1$, $r_{\lfloor (n+3f)/2 \rfloor+1-f}\geq r_{3f+1}$ and thus
    \begin{equation*}
        v_{k+1}^{p+1} = \frac{r_{f+1}+r_{\lfloor (n+3f)/2 \rfloor+1-f}}{2} \geq \frac{w_1^p+a_k^p}{2} = a_{k+1}^p.
    \end{equation*}
    
    Symmetrically, at most $h-(2f+1)$ fault-free states from $V^{(p)}$ fall in $V_3$ (i.e., greater than $A_k^p$). Therefore, $r_{\lfloor (n+3f)/2 \rfloor+1-f}\leq w_h^p$ and $r_{f+1}\leq A_k^p$,\footnote{At most $h-(3f+1)$ states are greater than $A_k^p$, so the $((n-f)-[h-(3f+1)])$-th state is less than or equal to $A_k^p$. Since $n\geq 5f+1$ and $h\leq n$, $(n-f)-[h-(3f+1)]\geq f+1$. This implies $r_{f+1}\leq r_{(n-f)-[h-(3f+1)]}\leq A_k^p$.} and thus $v_{k+1}^{p+1} \leq A_{k+1}^p$. This proves Equation (\ref{eq:BAC-thm-agreement}) for $k+1$.
\end{proof}

\begin{theorem}
    Algorithm DBAC satisfies $\epsilon$-agreement.
\end{theorem}

\begin{proof}
    Lemma \ref{lem:BACR-agreement} implies that $\interval(V^{(p+1)})\subseteq [a_n^p,A_n^p]$. Furthermore, by definition and geometric series, we have
    \begin{align*}
        &a_n^p=w_{1}^p+2^{-n}(w_{2f+1}^p-w_1^p), ~~~\text{and}\\~~~ &A_n^p=w_{h}^p+2^{-n}(w_{2f+1}^p-w_h^p).
    \end{align*}
    and  $w_1^p=\min V^{(p)}$ and $w_h^p=\max V^{(p)}$.
    Hence, $A_n^p-a_n^p$ gives an upper bound of convergence rate, which is $1-2^{-n}$ because $\range(V^{(p+1)})\leq (1-2^{-n})\range(V^{(p)})$.
    
    In conclusion,  DBAC satisfies $\epsilon$-agreement at $p_{\mathrm{end}}$, where
    \begin{equation}
        \label{eq:p-end-byz}
        p_{\mathrm{end}} = \frac{\log\epsilon}{\log(1-2^{-n})},
    \end{equation}
    because for all $p\geq p_{\mathrm{end}}$,
    $\range(V^{(p)})\leq (1-2^{-n})^{p_{\pend}}\leq \epsilon$.
\end{proof}

\section{Impossibility Results}
\label{s:impossible}

\subsection{Exact Consensus}

We first show that \ddegree{1, n-2} is not sufficient for solving binary exact consensus, even when all nodes are fault-free. Gafni and Losa prove the following theorem for a complete graph in a recent paper \cite{Gafni_TimeNotHealer_SSS23}:\footnote{As noted in \cite{Gafni_TimeNotHealer_SSS23}, the following theorem is different from, although similar to, the main result from the seminal paper by Santoro and Widmaye \cite{Santoro_Link}. The model in \cite{Santoro_Link} considers a synchronous system in which one node fails to send some of its messages per round.}

\begin{theorem}[\textbf{from \cite{Gafni_TimeNotHealer_SSS23}}]
Consider a synchronous model where in every round, each node might fail to receive one of the messages sent to it. It is impossible to achieve deterministic binary exact consensus, even when all nodes are fault-free.
\end{theorem}

This theorem implies the following corollary, since by the definition of \ddegree{1, n-2}, the message adversary can force any node to drop any single message sent to it.

\begin{corollary}
It is impossible to achieve deterministic binary exact consensus in the anonymous dynamic network with \ddegree{1, n-2}, even when all nodes are fault-free.    
\end{corollary}

\subsection{Crash-tolerant Approximate Consensus}


\begin{theorem}
\ddegree{T, \lfloor \frac{n}{2}\rfloor} and $n \geq 2f+1$ are together necessary for solving deterministic crash-tolerant approximate consensus.
\end{theorem}


\begin{proof}[Proof Sketch]
The proof consists of two parts. First, we show that it is impossible to achieve deterministic approximate consensus in an anonymous dynamic network with \ddegree{1, \lfloor \frac{n}{2}\rfloor-1}, even when all nodes are fault-free.    
The proof for the first part is by contradiction. Assume that there exists a deterministic approximate consensus algorithm in a dynamic graph with \ddegree{1, \lfloor \frac{n}{2}\rfloor-1}. However, to satisfy termination, a node $i$ must be able to make decision after communicating with only $\lfloor \frac{n}{2}\rfloor $ nodes (including $i$ itself). Therefore, it is possible for the message adversary to pick $\E(t)$ in a way that there are two non-overlapping groups of nodes that do not communicate with each other, making $\epsilon$-agreement impossible when these two group of nodes have different inputs. 

Second, we show that there exists a finite $T'$ such that it is impossible to achieve deterministic approximate consensus in an anonymous dynamic network with \ddegree{T', n-1} and $n \leq 2f$. The reason that we cannot prove for a general $T$ in this case is that it is trivial to design an algorithm that works for a fixed number of $T$, as each node can simply repeat the same process for $T$ rounds. However, we argue that it is impossible to do so with an unknown $T$. 


Assume that there exists a deterministic approximate consensus algorithm $\mathbb{A}$ in a dynamic graph with $n \leq 2f$ and \ddegree{T, n-1} for a fixed $T$. Now, we need to find a $T'$ so that $\mathbb{A}$ is incorrect, deriving a contradiction. 

Observe that to satisfy termination after $f$ nodes crash before the execution of the algorithm, a node $i$ must be able to make decision after communicating with only $\leq f$ nodes, since $n \leq 2f$. Without loss of generality, assume that in this scenario, namely Scenario 1, $\mathbb{A}$ takes $R$ rounds. 

Next we show that by choosing $T' = R+1$, $\mathbb{A}$ is incorrect. Consider the graph with \ddegree{R+1, n-1}. The message adversary can then pick $\E(t)$ for $1\leq t \leq R$ so that there are two non-overlapping groups of nodes that do not communicate with each other. This is possible because (i) this scenario is indistinguishable from Scenario 1, so nodes must output in $R$ rounds; (ii) within $R$ rounds, nodes only communicate with $\leq f$ nodes in $\mathbb{A}$; and (iii) $n \leq 2f$. 

Finally, consider the execution where these two non-overlapping groups are given different input value $0$ and $1$, respectively. Since each group makes a decision without communicating with another group, $\epsilon$-agreement is violated in $\mathbb{A}$, a contradiction.
\end{proof}

\subsection{Byzantine Approximate Consensus}


\begin{theorem}
    \ddegree{1, \lfloor \frac{n+3f}{2} \rfloor} and $n \geq 5f$ are together necessary for solving deterministic Byzantine approximate consensus.
\end{theorem}

\begin{proof}
The proof also consists of two parts. First, we show that it is impossible to achieve deterministic approximate consensus in the anonymous dynamic network with \ddegree{1, \lfloor \frac{n+3f}{2} \rfloor - 1} and $n \geq 3f+1$.\footnote{The necessity of $n \leq 3f$ is from prior work \cite{Lynch_DistributedAlgorithms1996,welch_book}.} Assume that there exists a deterministic Byzantine approximate consensus algorithm $\mathbb{A}$. 

Observe that to satisfy termination, a node $i$ must be able to make decision after communicating with only $\lfloor \frac{n+3f}{2} \rfloor$ nodes, including $i$ itself. Consider the following scenario:

\begin{itemize}
    \item Divide nodes into two groups: group A contains node $i_{1}, i_2, \dots, i_{\lfloor \frac{n+3f}{2} \rfloor}$, and group B contains $i_{\lfloor \frac{n-3f}{2} \rfloor+1}, \dots, i_{n}$. Note that each group has size $\lfloor \frac{n+3f}{2} \rfloor$, and there are $3f$ nodes in the intersection of the two groups.
    

    

    



    \vspace{5pt}

    \item Nodes $i_{\lfloor  \frac{n-f}{2} \rfloor+1}, \dots i_{\lfloor \frac{n+f}{2} \rfloor}$ are Byzantine faulty.

    \vspace{5pt}

    \item The message adversary picks $\E(t)$ in a way that nodes in group A receive only messages from group A and nodes in group B receive only messages from group B. 

    \vspace{5pt}

    \item Nodes $i_1, \dots, i_{\lfloor \frac{n-f}{2} \rfloor }$ have input $0$.

    \vspace{5pt}

    \item Nodes $i_{\lfloor \frac{n+f}{2} \rfloor+1}, \dots, i_n$ have input $1$.


    \item Byzantine nodes have the following behavior:  they behave to group A, as if they had input $0$; and behave to group $B$, as if they had input $1$. This is possible because of the anonymity assumption. Because the port numbering is potentially different at each node, Byzantine nodes have the ability to equivocate without being caught. In other words, useful primitives like reliable broadcast is impossible. 
\end{itemize}

Now, we make the following observations:

\begin{itemize}


    \item From the perspective of group A, only nodes $i_{\lfloor \frac{n+f}{2} \rfloor + 1}, \dots, i_{\lfloor \frac{n+3f}{2} \rfloor}$ have input $1$. The rest have input $0$. Note that the number of nodes with input $1$ in this case is exactly $f$.

    \item From the perspective of group B, only nodes $i_{\lfloor \frac{n-3f}{2} \rfloor+1}, \dots, i_{\lfloor \frac{n-f}{2} \rfloor}$ have input $0$. The rest have input $1$. Note that the number of nodes with input $0$ in this case is exactly $f$.
\end{itemize}

To satisfies validity, the first observation forces nodes in group A to output $0$, because there are only $f$ nodes with input $1$ and all of them could be Byzantine faulty. If group A outputs $1$ in this scenario, then it is straightforward to construct an indistinguishable scenario such that only nodes $i_{\lfloor \frac{n+f}{2} \rfloor + 1}, \dots, i_{\lfloor \frac{n+3f}{2} \rfloor}$ are Byzantine faulty (and nodes $i_{\lfloor  \frac{n-f}{2} \rfloor+1}, \dots i_{\lfloor \frac{n+f}{2} \rfloor}$ are fault-free), causing a violation of validity. 

Similarly, group B must output $1$, violating the $\epsilon$-agreement property. 

The second part of proving that there exists a finite $T'$ such that it is impossible to achieve deterministic Byzantine approximate consensus in the anonymous dynamic network with $n \leq 5f$ and \ddegree{T', n-1} is similar to the case of crash faults. We omit it here for lack of space. 
\end{proof}



\section{Conclusion and Discussion}
\label{s:discussion}

In conclusion, we study the feasibility of fault-tolerant consensus in anonymous dynamic networks. We identify the necessary and sufficient conditions for solving crash-tolerant and Byzantine approximate consensus. 


There are many interesting open problems in our model: 

\begin{itemize}
    \item We assume that nodes do not know the base graph $G$. Does knowing the graph help? 

    \item In practical applications, nodes might only know the IDs for a small set of other nodes. Does this knowledge help in increasing resilience or reducing the requirement for dynamic degree?

    \item What is the optimal convergence rate for Byzantine approximate consensus algorithms? 

    \item Observe that both our algorithms complete in $T\cdot p_{end}$ rounds in the worst case. For practical applications, it is useful to assume a probabilistic message adversary that picks $\E(t)$ randomly and investigate algorithms achieving the optimal expected number of rounds. 

    \item With unlimited bandwidth, one can indeed  simulate the algorithm in \cite{DolevLPSW86} by piggybacking the entire history of each node's past messages when sending the current state value. This achieves convergence rate of $1/2$. In fact, DBAC can improve the convergence rate by piggybacking a limited set of old messages, under limited bandwidth. It is interesting to identify the trade-off between bandwidth and convergence rate.
\end{itemize}



\end{document}